\documentclass[letterpaper,12pt]{article}
\usepackage{enumerate} 
\usepackage[margin=1in]{geometry}  
\usepackage{comment}
\usepackage{listings}
\usepackage{amsmath}
\usepackage{amsthm}
\usepackage{tikz}
\usepackage{caption}
\usepackage{array}
\usepackage{mdwmath}
\usepackage{multirow}
\usepackage{mdwtab}
\usepackage{eqparbox}
\usepackage{multicol}
\usepackage{amsfonts}
\usepackage{tikz}
\usepackage{multirow,bigstrut,threeparttable}
\usepackage{amsthm}
\usepackage{array}
\usepackage{bbm}
\usepackage{subfigure}
\usepackage{epstopdf}
\usepackage{mdwmath}
\usepackage{mdwtab}
\usepackage{eqparbox}
\usepackage{tikz}
\usepackage{latexsym}
\usepackage{cite}
\usepackage{amssymb}
\usepackage{bm}
\usepackage{amssymb}
\usepackage{graphicx}
\usepackage{mathrsfs}
\usepackage{epsfig}
\usepackage{psfrag}
\usepackage{setspace}
\usepackage{indentfirst}
\usepackage[
            CJKbookmarks=true,
            bookmarksnumbered=true,
            bookmarksopen=true,
            colorlinks=true,
            citecolor=red,
            linkcolor=blue,
            anchorcolor=red,
            urlcolor=blue
            ]{hyperref}
\usepackage{algorithm}
\usepackage{algpseudocode}
\usepackage{stfloats}
\usepackage{algorithm}
\usepackage{mathrsfs}
\usepackage{graphicx}
\usepackage{tabularx,booktabs,siunitx,caption}
\newcolumntype{C}{>{\centering\arraybackslash}X}

\theoremstyle{plain}
\newtheorem{theorem}{Theorem}

\def \bE {\mathbb{E}}

\usepackage{xspace}

\newcommand{\Prob}{\mathbb{P}}

\definecolor{myblue}{rgb}{.8, .8, 1}
\definecolor{mathblue}{rgb}{0.2472, 0.24, 0.6} 
\definecolor{mathred}{rgb}{0.6, 0.24, 0.442893}
\definecolor{mathyellow}{rgb}{0.6, 0.547014, 0.24}

\newcommand{\barS}{{\bar{S}}}

\usepackage{cleveref}
\crefname{lemma}{Lemma}{Lemmas}
\Crefname{lemma}{Lemma}{Lemmas}
\crefname{thm}{Theorem}{Theorems}
\Crefname{thm}{Theorem}{Theorems}
\crefformat{equation}{(#2#1#3)}

\begin{document}

\title{A Fast MCMC for the Uniform Sampling of Binary Matrices with Fixed Margins}
\author{Guanyang Wang\thanks{Guanyang Wang is with the Department of Mathematics, Stanford University. Email: \{guanyang\}@stanford.edu.}}

\maketitle

%
%
%
%
%
%
\begin{abstract}
	Uniform sampling of binary matrix with fixed margins is an important and difficult problem in statistics, computer science, ecology and so on.  The well-known swap algorithm would be inefficient when the size of the matrix becomes large or when the matrix is too sparse/dense. 
	Here we propose the Rectangle Loop algorithm, a Markov chain Monte Carlo algorithm to sample binary matrices with fixed margins uniformly. Theoretically the Rectangle Loop algorithm is better than the swap algorithm in Peskun's order. Empirically studies also demonstrates the Rectangle Loop algorithm is remarkablely more efficient than the swap algorithm.
\end{abstract}
\section{Introduction}\label{sec:Introduction}
The problem of sampling binary matrices with fixed row and column sums has attracted much attention in numerical ecology. In ecological studies, the binary matrix is called \textit{occurrence matrix}. Rows usually corresponds to species, the columns, to locations. For example, the binary matrix shown on Table \ref{tab:finch} is known as ``Darwin's Finch'' dataset, which comes from Darwin's studies of the finches on the Galapagos islands (an archipelago in the East Pacific). The matrix represents the presence/absence of 13 species of finches in 17 islands. A $``1"$ or $``0"$ in entry $(i,j)$ indicates the presence or absence of species $i$ at island $j$. It is clear from Table \ref{tab:finch} that some pairs of species tend to occur together (for example, species 9 and 10) while some other pairs tend to be disjoint. Therefore, it is of our interest to investigate whether the cooperation/competition influences the distribution of species on islands, or the patterns found are just by chance.
\begin{table}\caption{Occurrence Matrix occurrence matrix of the finches  on the Galapagos islands.}
\label{tab:finch}
\begin{center}
\begin{tabular}{|c|l|l|l|l|l|l|l|l|l|l|l|l|l|l|l|l|l|}
\hline
\multicolumn{1}{|l|}{}      & \multicolumn{17}{c|}{Island}                                      \\ \hline
\multicolumn{1}{|l|}{Finch} & A & B & C & D & E & F & G & H & I & J & K & L & M & N & O & P & Q \\ \hline
1                           & 0 & 0 & 1 & 1 & 1 & 1 & 1 & 1 & 1 & 1 & 0 & 1 & 1 & 1 & 1 & 1 & 1 \\ \hline
2                           & 1 & 1 & 1 & 1 & 1 & 1 & 1 & 1 & 1 & 1 & 0 & 1 & 0 & 1 & 1 & 0 & 0 \\ \hline
3                           & 1 & 1 & 1 & 1 & 1 & 1 & 1 & 1 & 1 & 1 & 1 & 1 & 0 & 1 & 1 & 0 & 0 \\ \hline
4                           & 0 & 0 & 1 & 1 & 1 & 0 & 0 & 1 & 0 & 1 & 0 & 1 & 1 & 0 & 1 & 1 & 1 \\ \hline
5                           & 1 & 1 & 1 & 0 & 1 & 1 & 1 & 1 & 1 & 1 & 0 & 1 & 0 & 1 & 1 & 0 & 0 \\ \hline
6                           & 0 & 0 & 0 & 0 & 0 & 0 & 0 & 0 & 0 & 0 & 1 & 0 & 1 & 0 & 0 & 0 & 0 \\ \hline
7                           & 0 & 0 & 1 & 1 & 1 & 1 & 1 & 1 & 1 & 0 & 0 & 1 & 0 & 1 & 1 & 0 & 0 \\ \hline
8                           & 0 & 0 & 0 & 0 & 0 & 0 & 0 & 0 & 0 & 0 & 1 & 0 & 0 & 0 & 0 & 0 & 0 \\ \hline
9                           & 0 & 0 & 1 & 1 & 1 & 1 & 1 & 1 & 1 & 1 & 0 & 1 & 0 & 0 & 1 & 0 & 0 \\ \hline
10                          & 0 & 0 & 1 & 1 & 1 & 1 & 1 & 1 & 1 & 1 & 0 & 1 & 0 & 1 & 1 & 0 & 0 \\ \hline
11                          & 0 & 0 & 1 & 1 & 1 & 0 & 1 & 1 & 0 & 1 & 0 & 0 & 0 & 0 & 0 & 0 & 0 \\ \hline
12                          & 0 & 0 & 1 & 1 & 0 & 0 & 0 & 0 & 0 & 0 & 0 & 0 & 0 & 0 & 0 & 0 & 0 \\ \hline
13                          & 1 & 1 & 1 & 1 & 1 & 1 & 1 & 1 & 1 & 1 & 1 & 1 & 1 & 1 & 1 & 1 & 1 \\ \hline
\end{tabular}
\end{center}
\end{table}

Assuming different species have independent distributions on islands, then the observed binary matrix is simply a random sample from the uniform distribution of all the binary matrices with fixed margins. Table \ref{tab:fixed_sum} gives an example of all configurations of  $3\times 3$ binary matrices with $[1,2,1]$ as both row and column sums. Ideally, if we could list all the binary matrices with arbitrary size, then we could compare the pattern found in the observed matrix with others,   to conclude whether the observed matrix is simply by chance. However, enumerating matrices with fixed margins is often impractical both theoretically and computationally for moderate size of matrices. Therefore, sampling such random matrices becomes the natural choice. 

\begin{table}\caption{All possible $3\times 3$ binary matrices with $[1,2,1]$ as both row and column sums}
\label{tab:fixed_sum}
\begin{center}
\begin{tabular}{ c c c c c }
\toprule
A & B & C & D & E\\
\midrule\\
\addlinespace[-2ex]
$ \begin{bmatrix}  0 & 1 & 0  \\ 1 &  0 & 1 \\ 0 & 1 & 0 \end{bmatrix}$ & 
$ \begin{bmatrix}  0 & 1 & 0  \\ 1 &  1 & 0 \\ 0 & 0 & 1 \end{bmatrix}$ &
$ \begin{bmatrix}  1 & 0 & 0  \\ 0 &  1 & 1 \\ 0 & 1 & 0 \end{bmatrix}$ &
$ \begin{bmatrix}  0 & 1 & 0  \\ 0 &  1 & 1 \\ 1 & 0 & 0 \end{bmatrix}$ &
$ \begin{bmatrix}  0 & 0 & 1  \\ 1 &  1 & 0 \\ 0 & 1 & 0 \end{bmatrix}$ 
\\
\addlinespace[1.5ex]
\bottomrule
\end{tabular}
\end{center}
\end{table}

The problem of sampling such matrices also occurs in many other fields, with different names. For example, an equivalent formulation is uniformly sampling undirected bipartite graphs with given vertex degrees. A bipartite graph $G = (U,V,E)$ is a graph whose vertices are divided into two disjointed sets, denoted by $U = \{u_1, \cdots, u_m\}$, $V = \{v_1, \cdots, v_n\}$. $E$ is called the edge set where every edge connects one vertex in $U$ to one in $V$. The binary matrix $M = (m_{i,j})_{m\times n }$ is often called the  bi-adjacency matrix of $G$ and is defined by 
\[
m_{i,j} = \begin{cases}
1, \text{if there is an edge connecting}~ u_i ~\text{and}~ v_j,\\
0, \text{otherwise}.
\end{cases}
\]
Bipartite graphs are often used in network studies to model the interaction between two objects, for example, customers and products. It is often required to  sample graphs with preserved degree sequence in network analysis uniformly. Throughout this paper, we will use the term `binary matrix' instead of `bipartite graph' to avoid confusion, although they are equivalent.

The algorithms of sampling binary matrices with fixed margins are  divided into two classes. The first class of algorithms relies on the \textit{rejection sampling} or \textit{importance sampling} techniques, see \cite{snijders1991enumeration}, \cite{miklos2004randomization}, \cite{chen2005sequential}, \cite{harrison2013importance} \cite{holmes1996uniform} for examples. Importance sampling usually provides degree from non-uniform distribution, but it can be used to construct estimators to estimate the quantities of interest, such as the number of binary matrices with given margins. Chen et al. \cite{chen2005sequential} introduced a sequential importance sampling (SIS) scheme to test the hypothesis we mentioned at the beginning of this paper on ``Darwin's Finch" dataset. 

The second class falls into the \textit{Markov Chain Monte Carlo} (MCMC) category and will be our main focus in this paper. The well-known ``swap algorithm" has been used for decades. To the author's best knowledge, it is first introduced by Besag and Clifford \cite{besag1989generalized} in 1989 to solve a statistical testing problem. The swap algorithm has been formally proposed and analysed by Rao et al. \cite{rao1996markov} and \cite{kannan1999simple} in the 1990s. A similar question is to sample  matrices with non-negative integer entries, fixed row and column sums. Diaconis and Gangolli have proposed a random walk Metropolis algorithm \cite{diaconis1995rectangular}. Many variations and extensions of this algorithm are described by Diaconis and Sturmfels \cite{diaconis1998algebraic}.  

The swap method attempts to make a \textit{single} swap in each iteration, but when the matrix is large, or is mostly filled (or unfilled), the efficiency of swap algorithm can be relatively low. In 2008, Verhelst \cite{verhelst2008efficient} proposed a new MCMC algorithm based on the idea of performing multiple swaps per iteration. In 2014,  Strona et al. \cite{strona2014fast} introduced the ``Curveball algorithm", which uses a `fair trade' operation to replace the `swap' operation in the swap algorithm, aiming for a faster mixing. The mathematical formulation of Curveball algorithm is equivalent to Verhelst's algorithm, but with different implementation and reasoning. A nice survey and numerical comparisons of the existing algorithms can be found in a recent dissertation \cite{rechner2018markov}. The class of `multiple swaps' algorithms tends to improve the mixing time empiricially. However, each step of the `multiple swaps' algorithm is slower than the classical swap algorithm. Meanwhile, it is hard to compare the `multiple swaps' algorithms and classical swap algorithm theoretically, as the corresponding Markov-chains have complicated behaviors and are therefore  hard to analyze mathematically. The only existing result can be found in \cite{carstens2018speeding}.

In this paper, we introduce a novel algorithm called Rectangle Loop algorithm. The algorithm is based on the classical swap algorithm, with a careful utilization of the matrix structure given by margins. We have also proved the resulting Markov Chain dominates the classical chain used in the swap algorithm in the sense of Peskun's partial ordering \cite{peskun1973optimum}, and is easy to implement. Section \ref{sec:Existing Methods} gives a review of swap algorithm and Curveball algorithm, including the details of both algorithms and a discussion. In Section \ref{sec:rectangle_loop} we introduce our new algorithm --  Rectangle Loop algorithm. Section \ref{sec:theoretical_results} proves the theoretical properties of Rectangle Loop algorithm. Section \ref{sec: Simulation_results} gives numerical results.
\section{Existing Methods}\label{sec:Existing Methods}
\subsection{Swap Algorithm}
The swap algorithm, or equivalently, swap chain is based on the idea of  \textit{swapping checkerboard units}. Here a checkerboard unit is a two by two matrix with one of the following forms:
\[
\begin{pmatrix} 
1 & 0 \\
0 & 1 
\end{pmatrix},
\begin{pmatrix} 
0 & 1 \\
1 & 0
\end{pmatrix}.
\]
A swap means changing one checkerboard unit to the other.

Starting from an initial matrix, one chooses two rows and two columns uniformly at random among all rows and columns. If the resulting $2\times 2$ submatrix with entries in the intersection of these rows and columns is a checkerboard unit, it is swapped, otherwise, do nothing. 
\begin{algorithm}
\caption{Swap Algorithm}\label{alg:swap}
\hspace*{\algorithmicindent} \textbf{Input:} initial binary matrix $A_0$, number of iterations $T$ \\

\begin{algorithmic}[1]

\For{$t= 1,\cdots T$}
\State Choose two distinct rows and two distinct columns uniformly at random
\State \textbf{If}  the corresponding $2\times 2$ submatrix of $A_{t-1}$ is a checkerboard unit, i.e.
\[
\begin{pmatrix} 
1 & 0 \\
0 & 1 
\end{pmatrix} \qquad \text{or} \qquad
\begin{pmatrix} 
0 & 1 \\
1 & 0
\end{pmatrix} ,
\]

swap the submatrix
$\begin{pmatrix} 
1 & 0 \\
0 & 1 
\end{pmatrix}
\leftarrow \begin{pmatrix} 
0 & 1 \\
1 & 0
\end{pmatrix}$ or vice versa.

\textbf{Otherwise} $A_{t}  \leftarrow A_{t-1} $ 

\EndFor
\end{algorithmic}
\end{algorithm}

The swap algorithm is a Metropolis-type Markov chain Monte Carlo which converges to the uniform distribution. 

\subsection{Curveball Algorithm}

The swap algorithm can often be inefficient, taking Darwin's Finch data for example, there are ${13\choose2} {17\choose2} = 10608$ submatrices with size $2\times 2$, however, only about $3\%$ of them are swappable. This means it requires a very large $T$ (the number of iterations) to ensure the generated degree is close to uniformly distributed. The Curveball algorithm provides another solution.

\begin{algorithm}
\caption{Curveball Algorithm}\label{alg:Curveball}
\hspace*{\algorithmicindent} \textbf{Input:} initial binary matrix $A_0$, number of iterations $T$ \\

\begin{algorithmic}[1]
\For{$t= 1,\cdots T$}
\State Choose two distinct rows $r_a, r_b$ uniformly at random
\State Determine two disjoint sets 
\begin{align*}
 &S_{a-b}  \doteq \{k: A_t(a,k) = 1, A_t(b,k) = 0\} \\
 &S_{b-a} \doteq \{l: A_t(a,l) = 0, A_t(b,l) = 1\},\\
\end{align*}
here assuming $\lvert S_{i-j} \rvert \leq \lvert S_{j-i} \rvert$
 \State Choose a subset $V\subset S_{j-i}$ uniformly at random
 \State Set $A_{t+1} = A_{t}$ except for row $a, b$. 
 
 For row $a$:
 \[
 A_{t+1}(a, l) = \begin{cases}
1 \quad \text{if}~ l\in V\\
0 \quad \text{if}~ l\in S_{a-b} \cup S_{b-a} \setminus V \\
A(a,l) \quad \text{otherwise}
 \end{cases}
 \]
 
 For row $b$
 \[
A_{t+1} (b, k ) =  \begin{cases}
1 \quad \text{if}~ k\in S_{a-b} \cup S_{b-a} \setminus V\\
0 \quad \text{if}~ k\in  V \\
A(b,k) \quad \text{otherwise}
 \end{cases}
 \]
\EndFor
\end{algorithmic}
\end{algorithm}

The Curveball algorithm uses `trade' instead of `swap' operation in each iteration. Steps 3-5 in Algorithm \ref{alg:Curveball} gives an illustration of trading, it trades elements in column $V$ with elements in column $S_{a-b}$ for row $a$ and row $b$, preserving their row and column sums. Though seemingly complicated, there is a very intuitive explanation of the Curveball algorithm. We refer the readers to \cite{strona2014fast} for detailed illustrations.

\section{Rectangle Loop Algorithm}\label{sec:rectangle_loop}
The swap algorithm is proven to converge to uniform distribution. However, getting stuck at the same configuration is inefficient and thus the convergence could be very slow. For example, numerical experiments suggest that one would expect more than $30$ iterations before each successful swap using `Darwin's Finch' dataset. Assuming the randomly chosen row is mostly filled, such as row $1$ in Table \ref{tab:finch}, the two random chosen entries in this row would most likely be $[1,1]$ but the row of a `checkerboard unit' has to be either $[1,0]$ or $[0,1]$. Therefore swapping rarely happens when the chosen row/column is mostly filled (or equivalently, mostly unfilled).

The Rectangle Loop algorithm is designed to increase the chance of swapping. The idea is illustrated in Figure $1$. In this example the target matrix is of size $5\times 5$, with row names $R_1, \cdots, R_5$ and column names $C_1, \cdots, C_5$. In each step, we choose one row and one column uniformly at random (Step A). Suppose $R_2$ and $C_2$ is chosen, with corresponding entry $1$ , the red number in the top middle plot of Figure \ref{fig:Rectangle_loop}. Then we randomly choose a $0$ among all the $0$s in $R_2$ (Step B). Since there is only one $0$ in $R_2$, which is at location $C_4$, this is our only choice. Again, we scan through all the entries in the same column with the $0$ just chosen ($C_4$) and randomly choose a $1$ among all $1$s (Step C). In our example, the $1$s of $C_4$ are located at $R_1$ and $R_4$. Suppose we have chosen $(R_4, C_4)$. Now the three locations $(R_2, C_2), (R_2, C_4), (R_4, C_4)$ altogether give us the fourth one $(R_4, C_2)$, making the four entries a rectangle (Step D). If the fourth entry equals $0$, then we swap the submatrix as we did in swap method (Step E). Otherewise the fourth entries equals $1$ and the original matrix is not changed.  After Step A - E is iterated for many times, the resulting randomized matrices are used as representatives of uniformly distributed matrix with fixed margins.

The main difference between the Rectangle Loop algorithm and the swap algorithm is the sampling scheme. The Rectangle Loop algorithm is performing `conditional sampling', making it more efficient than swap method, which is doing `unconditional choosing'. For example, suppose both the swap method and Rectangle Loop algorithm have chosen $R_2$ and $C_2$, an entry with value $1$. Then $R_2$ has only one $0$ which is in column $4$. For the swap algorithm, the probability of correctly choosing $C_4$ is only $\frac 14$, as it is uniformly choosing among all columns.  The Rectangle Loop algorithm, however, as the mechanism guarantees we only sample from the zero entries, chooses $C_4$ with probability $1$. Therefore it significantly increase the swapping probability, leading to a faster convergence than the swap chain.

The details of the Rectangle Loop algorithm is described in Algorithm \ref{alg:Rectangle}. Noteworthy, when finding a $0$ entry, we sample $1$ with the same column as the $0$. When finding a $1$, we sample $0$ with the same row as the $1$. This `symmetric' design ensures the algorithm that converges to the correct distribution, as will be proved in Section \ref{sec:theoretical_results}. The paths of sampled entries in each iteration always form a rectangle, that is where the name `Rectangle Loop algorithm' comes from.

\begin{figure}[htbp]
\includegraphics[width= \textwidth]{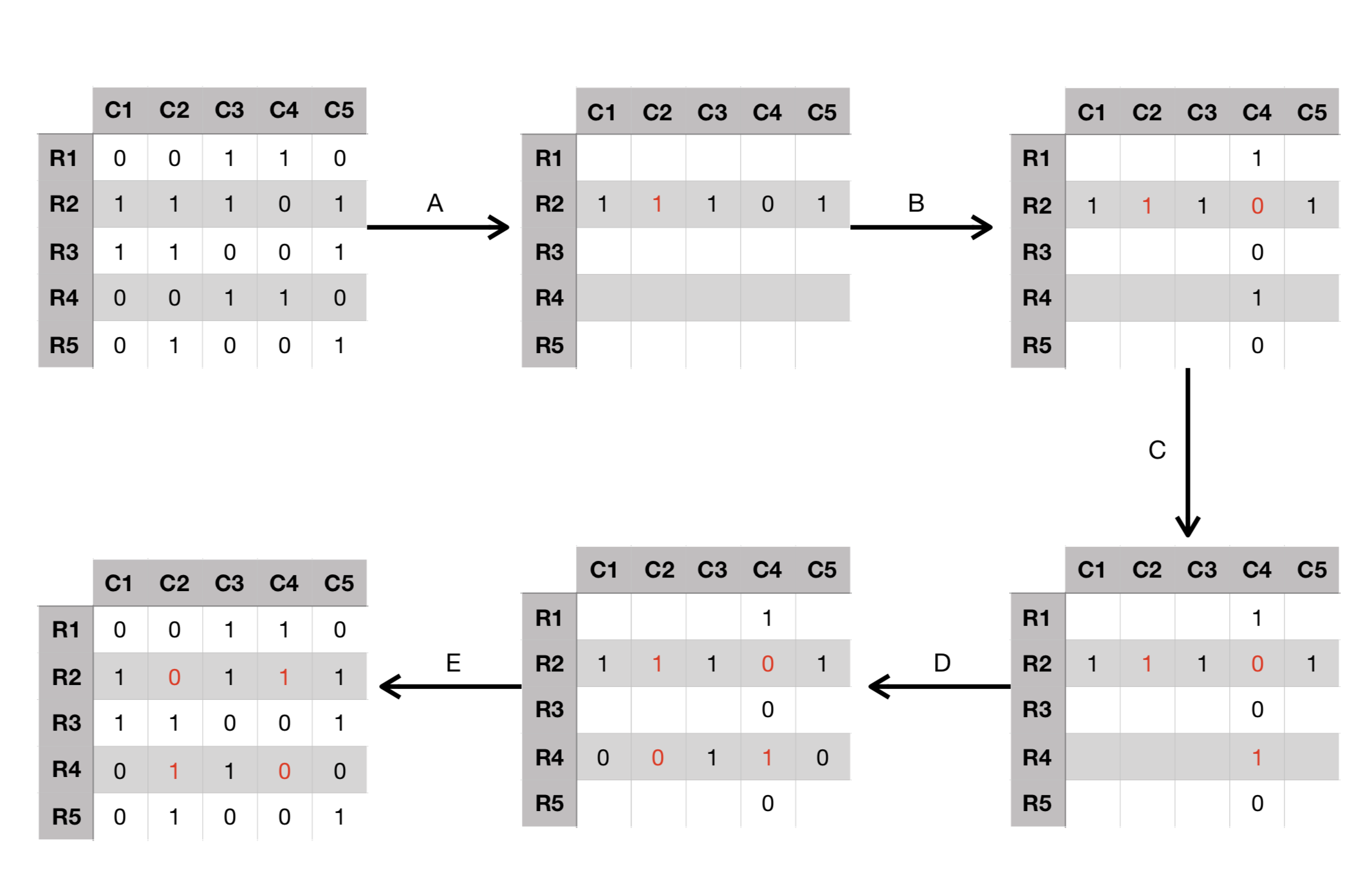}
\caption{An illustration of Rectangle Loop algorithm.}
\label{fig:Rectangle_loop}
\end{figure}

\begin{algorithm}
\caption{Rectangle Loop Algorithm}\label{alg:Rectangle}
\hspace*{\algorithmicindent} \textbf{Input:} initial binary matrix $A_0$, number of iterations $T$ \\

\begin{algorithmic}[1]
\For{$t= 1,\cdots T$}
\State Choose one row and one column $(r_1,c_1)$ uniformly at random
\If  {$A_{t-1}(r_1,c_1) = 1$}
\State Choose one column $c_2$ at random among all the $0$ entries in $r_1$
\State Choose one row $r_2$ at random among all the $1$ entries in $c_2$
\Else ~~ {$A_{t-1}(r_1,c_1) = 0$}
\State Choose one row $r_2$ at random among all the $1$ entries in $c_1$
\State Choose one column $c_2$ at random among all the $0$ entries in $r_2$
\EndIf
\If {the submatrix extracted from $r_1, r_2, c_1, c_2$ is a `checkerboard unit'}
\State Swap the submatrix
\Else ~~ {$A_t \leftarrow A_{t-1}$}
\EndIf
\EndFor
\end{algorithmic}
\end{algorithm}
\section{Theoretical Results}\label{sec:theoretical_results}
Given  row sums $\textbf{r} = (r_1, r_2, \cdots, r_m)$ and column sums $\textbf{c} = (c_1, c_2, \cdots, c_n)$, we define $\Sigma_{\textbf r, \textbf c}$ be the set of all matrices with row sums $\textbf r$ and column sums $\textbf c$. The suffcient and necessary condition for $\Sigma_{\textbf r, \textbf c}$ not being zero is given by Gale \cite{gale1957theorem}, Ryser \cite{ryser2009combinatorial} in 1957. We call two matrices $A$, $B$ `swappable' if one can transform to the other via one step swap algorithm. Equivalently, $A$ and $B$ only differs in a $2\times 2$ `checkerboard unit'.  For the sake of simplicity, we assume henceforth $0 < r_i < n, 0 < c_j < m$ for any $1\leq i \leq m, 1\leq j \leq n$, as otherwise we could simply delete that degenerate row/column. The following theorem characterizes the limit distribution and transition probability of the swap chain.
\begin{theorem}
Given $\textbf{r}, \textbf{c}$ and an initial matrix $A_0 \in \Sigma_{\textbf r, \textbf c}$, the swap algorithm defines a Metropolis-type Markov chain with stationary distribution $\text{Unif}(\Sigma_{\textbf r, \textbf c})$, transition kernel:
\[
\Prob( A , B) = \begin{cases}
 \frac{1}{{m \choose 2} {n \choose 2}} \qquad &\text{If $A$ and $B$ are swappable},\\ 
 1 - \frac{s(A)}{{m \choose 2} {n \choose 2}} \qquad &\text{If $B = A$}, s(A) \doteq \#\{\text{C: C and A are swappable}\}\\
 0 \qquad &\text{Otherwise}
\end{cases}
\]
and acceptance probability $1$.
\end{theorem}
\begin{proof}
When $A$ and $B$ are swappable, there exists two rows $i_1, i_2$ and two columns $j_1, j_2$ such that $A$ and $B$ only differs in the $2\times 2$ submatrix extracted by row $i_1, i_2$ and column $j_1, j_2$. Therefore the probability of swapping $A$ to $B$ equals
\[
\frac{1}{{m \choose 2} {n \choose 2}}.
\]
Recall that in the setting of Metropolis-Hastings, the acceptance probability from $A$ to $B$ is given by 
\[
\min\{1, \frac{\pi(B)\Prob(B,A)}{\pi(A)\Prob(A,B)},\}
\]
notice here the stationary distribution $\pi$ is designed to be $\text{Unif}(\Sigma_{\textbf r, \textbf c})$, $\Prob(A,B) = \Prob(B,A) = \frac{1}{{m \choose 2} {n \choose 2}}$. Hence it is clear that the swap algorithm is a Metropolis-type Markov chain with $\text{Unif}(\Sigma_{\textbf r, \textbf c})$ as stationary distribtuion and acceptance probability $1$, which justifies the correctness of the swap algorithm.
\end{proof}
The key point in swap algorithm is symmetry. When two different states $A,B$ are swappable, the associated transition probability is symmetric, i.e., 
\[
\Prob(A,B) = \Prob(B,A),
\]
this ensures the chain has acceptance probability $1$.  

In order to compare the efficiency of different Markov kernels with the same distribution, Peskun \cite{peskun1973optimum} first introduced the following partial-ordering. 

Let $\Prob_1$, $\Prob_2$ be two Markov transition kernels on the same state space $\mathcal S$ with same stationary distribution $\pi$, then $\Prob_1$ \textit{dominates} $\Prob_2$ \textit{off the diagonal}, $\Prob_1 \succeq \Prob_2 $, if 
\[
\Prob_1(x, A) \geq \Prob_2 (x, A)
 \]
 for all $x\in \mathcal S$ and $A$ measurable with $x\notin A $.
 
 When the state space is finite, as in our case, $\Prob_1 \succeq \Prob_2$ iff all the off-diagonal entries of $\Prob_1$ are greater than or equal to the corresponding off-diagonal entries of $\Prob_2$. This indicates $\Prob_1$ has lower probability to get stuck in the same state, and is exploring the state space in a more efficient way. The following theorem shows the rectangale loop algorithm also has uniform distribution as stationary distribution, and the corresponding chain dominates the swap chain off the diagonal. For simplicity, we will use $\Prob_s$ and $\Prob_r$ to denote transition kernel for the swap chain and rectangle loop chain, respectively, omitting its dependency on
 $\textbf{r},  \textbf{c}, \text{Unif}(\Sigma_{\textbf r, \textbf c}).$ 
 
 \begin{theorem}
 	Given $\textbf{r}, \textbf{c}$ and an initial matrix $A_0 \in \Sigma_{\textbf r, \textbf c}$, the Rectangale loop algorithm defines a Metropolis-type Markov chain with stationary distribution $\text{Unif}(\Sigma_{\textbf r, \textbf c})$. The transition kernel $\Prob_r$ dominates $\Prob_s$ off the diagonal.
 \end{theorem}

\begin{proof}
Given any two swappable configurations $A$ and $B$, we are aiming to show $\Prob_r(A, B) = \Prob_r(B,A) \geq \Prob_s(A, B)$.

As $A, B$ are swappable, there exists two rows $i_1, i_2$ and two columns $j_1, j_2$ such that $A$ and $B$ only differs in the $2\times 2$ submatrix extracted by row $i_1, i_2$ and column $j_1, j_2$. Without loss of generality, we assume the checkerboard unit corresponding to $A$ has the form $\begin{bmatrix}
1 & 0 \\ 0 & 1
\end{bmatrix}$, as shown in Figure \ref{sec:theoretical_results}. Notice that there are four vertices of the $2\times 2$ submatrix and the Rectangle Loop algorithm chooses one arbitray row and column at its first step. This suggests the probability of transforming $A$ to $B$ equals the summation of four probabilities, each one corresponds to choosing one specific vertex of the `checkerboard unit' .

Figure \ref{fig:prob_A_B} illustrates the calculation of one path, starting from the vertex $(i_1, j_1)$. The possibility of choosing row $i_1$ and column $j_1$ is $\frac 1 {mn}$. Then one chooses a $0$ among all the $0$s in row $i_1$, and there are $n - r_{i_1}$ of them. Therefore the possibility of choosing column $j_2$ equals $\frac{1}{n - r_{i_1}}$. Similarly, after choosing $j_2$, one chooses a $1$ among all $1$s in column $j_2$, and there are $c_{j_2}$ of them. Hence the possibility of choosing row $i_2$ equals $\frac 1 {c_{j_2}}$. The fourth entry is fixed after determining the first three entries thus the last step has probability $1$. Multipling the possibilities above altogether,  the possibility of transforming $A$ to $B$, starting with $(i_1, j_1)$, equals
\[
\frac{1}{mn}\cdot \frac{1}{n - r_{i_1}} \cdot \frac 1 {c_{j_2}}.
\]
\begin{figure}[htbp]
	\includegraphics[width= \textwidth]{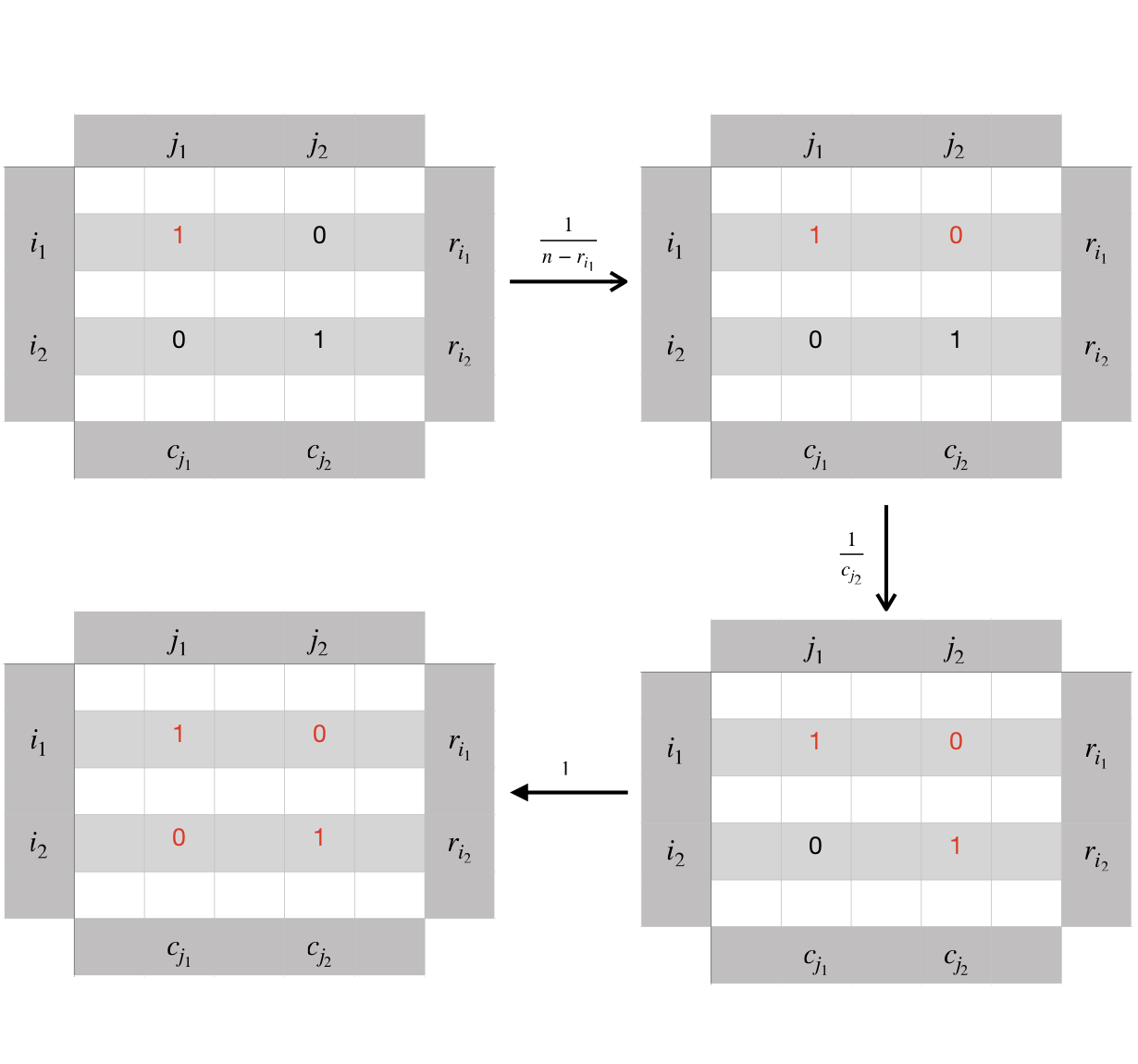}
	\caption{An illustration of calculating $\Prob_r(A,B)$ for a single path, starting from vertex $(i_1, j_1)$.}
	\label{fig:prob_A_B}
\end{figure}
The probability of transforming $A$ to $B$ with other starting vertex can be calculated accordingly. It turns out $\Prob_r(A,B)$ can be written as the following summation:
\[
\Prob_r(A,B) = \frac{1}{mn} \bigg( \frac{1}{n - r_{i_1}} \cdot \frac 1 {c_{j_2}} +  \frac 1 {c_{j_2}} \cdot  \frac{1}{n - r_{i_2}} +   \frac{1}{n - r_{i_2}} \cdot \frac 1 {c_{j_1}} +  \frac 1 {c_{j_1}} \cdot  \frac{1}{n - r_{i_1}} \bigg).
\]

Following the same strategy,  $\Prob_r(B,A)$ can also be calculated below. Figure \ref{fig:prob_B_A} illustrates the calculation of  $\Prob_r(B,A)$  starting from $(i_1, j_1)$.
\[
\Prob_r(B,A) = \frac{1}{mn} \bigg( \frac 1 {c_{j_1}} \cdot  \frac{1}{n - r_{i_2}} +   \frac{1}{n - r_{i_2}} \cdot \frac 1 {c_{j_2}} +   \frac 1 {c_{j_2}}\cdot \frac{1}{n - r_{i_1}} +   \frac{1}{n - r_{i_1}} \cdot \frac 1 {c_{j_1}} \bigg).
\]
\begin{figure}[htbp]
	\includegraphics[width= \textwidth]{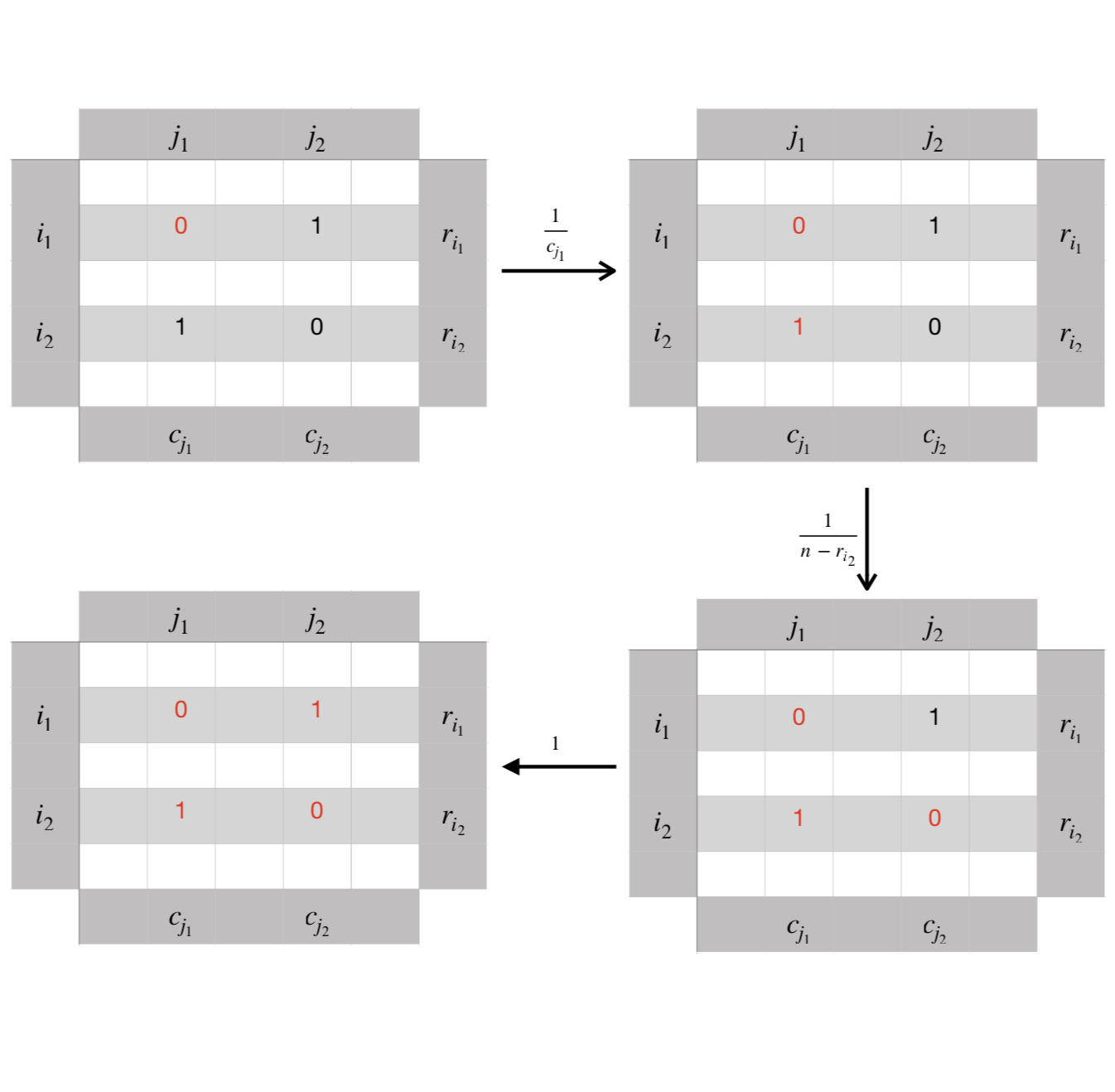}
	\caption{An illustration of calculating $\Prob_r(B,A)$ for a single path, starting from vertex $(i_1, j_1)$.}
	\label{fig:prob_B_A}
\end{figure}
After matching all the terms of $\Prob_r(B,A) $ with $\Prob_r(A,B)$, we conclude that  $\Prob_r(B,A)  = \Prob_r(A,B)$, which justifies the Rectangle Loop algorithm has Unif($\Sigma_{\textbf r, \textbf c}$) as stationary distribution.

To show $\Prob_r(A, B) \geq \Prob_s(A, B)$,  notice that $\Prob_s(A, B) =  \frac{1}{{m \choose 2} {n \choose 2}} = \frac 4 {m(m-1)n(n-1)}$. and 
\begin{align*}
\Prob_r(A,B) = &\frac{1}{mn} \cdot  \frac{1}{n - r_{i_1}} \cdot \frac 1 {c_{j_2}} + \frac{1}{mn} \cdot  \frac 1 {c_{j_2}} \cdot  \frac{1}{n - r_{i_2}} + \\
&  \frac{1}{mn} \cdot  \frac{1}{n - r_{i_2}} \cdot \frac 1 {c_{j_1}} + \frac{1}{mn} \cdot  \frac 1 {c_{j_1}} \cdot  \frac{1}{n - r_{i_1}} .\\
\end{align*}
It is clear that $\Prob_r(A,B)$ can be written as the summation of four terms. Each term in the summation is greater than or equal to $\frac 1 {m(m-1)n(n-1)}$. Therefore we conclude that $\Prob_r(A,B) \geq \Prob_s(A,B) $  for any swappable $A,B$. This indicates $\Prob_r \succeq \Prob_s$, the Rectangle Loop algorithm is exploring the state space in a more efficient way than the swap algorithm.
\end{proof}
\section{Simulation Results and Applications}\label{sec: Simulation_results}
\subsection{A  Concrete Example}
Now we use the example used in \cite{strona2014fast} and \cite{miklos2004randomization} to compare the existing algorithms and the Rectangle Loop algorithm. The example below is concrete. The transition matrix can be calculated explicitly and convergence can be assessed analytically.

The five matrices shown in Table \ref{tab:fixed_sum} are all possible configurations of  $3\times 3$ binary matrices with $[1,2,1]$ as both row and column sums. 
The transition matrices for swap algorithm, Curveball algorithm and Rectangale loop algorithm is shown in 
Table \ref{tab:transition}.

\begin{table}[htbp]
	
	\begin{center}
		\begin{tabular}{ c c c  }
			\toprule
			Swap & Curveball & Rectangle Loop \\
			\midrule\\
			\addlinespace[-2ex]
			$ \begin{bmatrix}  \frac5 9 & \frac19 & \frac19 & \frac19 &\frac19 \\ \frac 19&  \frac23 & 0 &\frac19 & \frac 19 \\ \frac19& 0 & \frac 23 &\frac19 &\frac19 \\ \frac 19 & \frac 19 & \frac19 &\frac 23 & 0\\
			\frac19 & \frac19 &\frac 19 &0 & \frac23
			\end{bmatrix}$ & 
			$ \begin{bmatrix}  
			\frac13 & \frac16 & \frac16 & \frac16 &\frac16  \\ \frac16 &  \frac13 & 0  & \frac16 &\frac13\\ \frac16 & 0 & \frac13 &\frac13 & \frac 16\\
			\frac 16 & \frac16 &\frac13 &\frac13 & 0 \\
			\frac 16 &\frac 13 &\frac16 &0 &\frac 13			 \end{bmatrix}$ &
			$ \begin{bmatrix}  0 & \frac14 & \frac14 &\frac14 &\frac14  \\ \frac14 &  \frac14 & 0 &\frac13 &\frac16 \\ \frac14& 0 & \frac14 & \frac 16 &\frac 13\\
			\frac 14 &\frac 13 &\frac16 &\frac 14 & 0\\
			\frac14 &\frac16 &\frac 13 &0 &\frac14			 \end{bmatrix}$ 
			\\
			\addlinespace[1.5ex]
			\bottomrule
		\end{tabular}
	\end{center}
	\caption{Transition matrices for swap algorithm (left), Curveball algorithm (middle), Rectangle Loop algorithm (right)}
	\label{tab:transition}
\end{table}
Figure \ref{fig:comparison_3by3} shows the comparison of the three algorithms. Here we measure the distance between transition kernel $\Prob$ and the stationary distribution $\pi$ by total variation distance:

\[
\max_{A\in \Sigma_{\textbf r, \textbf c}}\lVert \Prob^k(A,\cdot) - \pi \rVert_{\text{TV}} = \frac 12 \max_{A\in \Sigma_{\textbf r, \textbf c}} \sum_{B\in \Sigma_{\textbf r, \textbf c}}\lvert\Prob^k(A,B) - \pi(B) \rvert,
\]
where $k$ denotes the power.
It is clear that all the algorithms converge exponentially to uniform distribution, but the Rectangle Loop algorithm converges faster than swap algorithm and Curveball algorithm.
\begin{figure}[htbp]
	\includegraphics[width= \textwidth]{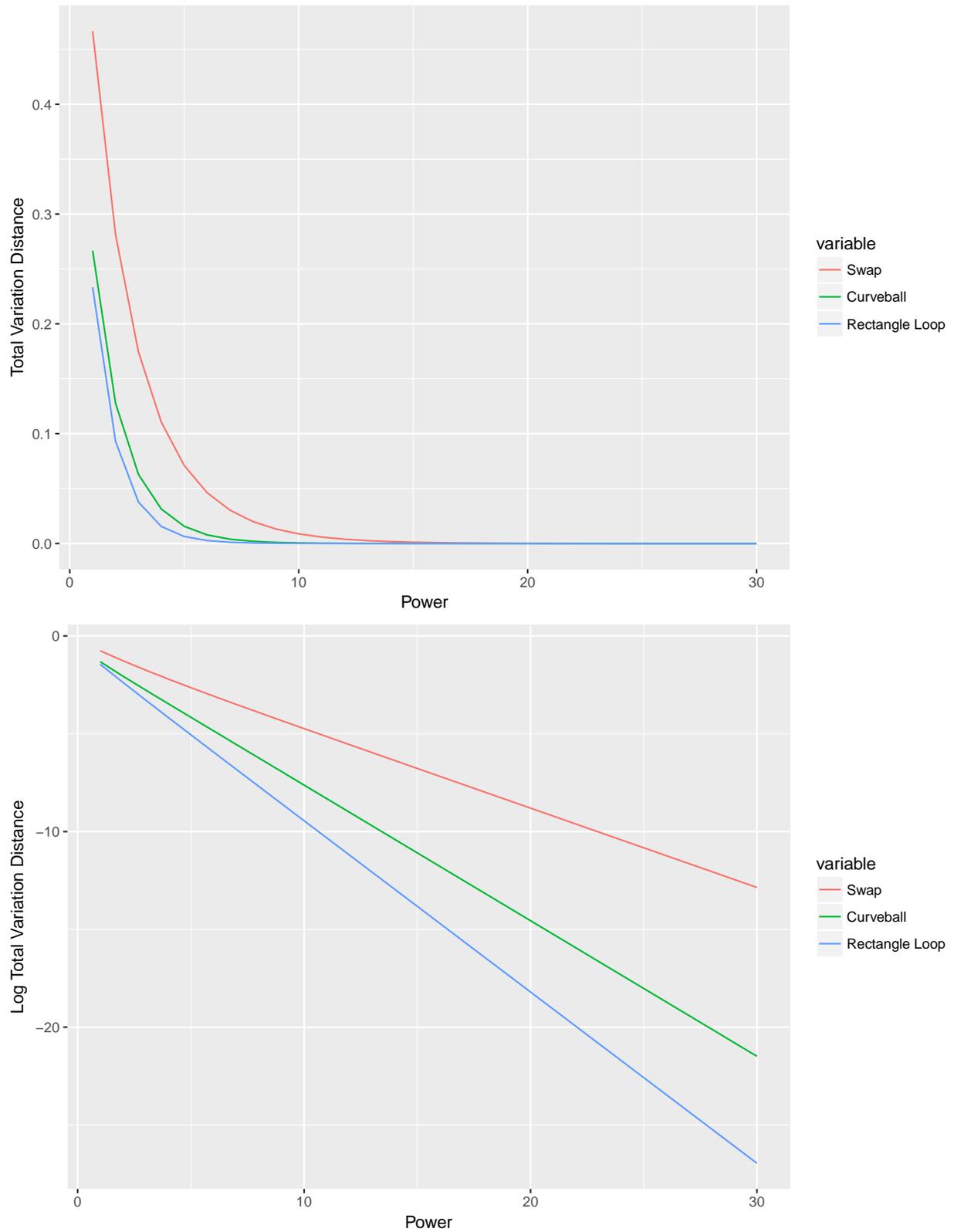}
	\caption{Comparison between swap algorithm, Curveball algorithm and Rectangle Loop algorithm. Top: relationship between total variation distance and the power of transition matrix. Bottom: relationship between logarithm of total variation distance and the power of transition matrix. }
	\label{fig:comparison_3by3}
\end{figure}
\subsection{Experiments on Empirical Mixing Time}
For larger matrices, it is infeasible to calculate the transition matrix theoretically. To provide empirical justification for the advantage of the Rectangle Loop algorithm. We have designed the experiment as follows. For each $p = 0.01, 0.05, 0.1, 0.2, 0.3, 0.4, 0.5$, a $100\times 100$ binary matrix is generated for which each entry has probability $p$ to be $1$. We ran each algorithm for $10000$ iterations and collected the corresponding number of swaps, as shown in Table \ref{tab:comparison_rectangle_swap}. When the filled portion $p$ is small, the Rectangle Loop algorithm is extremely efficient, producing more than $73$ times more swaps than the swap algorithm. For large $p$, the advantage of Rectangle Loop algorithm is reduced, but still very significant. For $p = 0.5$, the Rectangle Loop still produces $4$ times more swaps than the swap algorithm. Noteworthy, the zeros and ones play the symmetric rule in a binary matrix, therefore it is not necessary to generate the random matrix for $p > 0.5$.

The result above justifies our theoretical result that the Rectangle Loop algorithm converges faster than the swap algorithm. However, the above experiments did not consider the running time for each iteration. In fact, one iteration of Rectangle Loop algorithm is computationally more expensive than that of swap algorithm. To investigate this issue, we also record the time per swap for both algorithms, as shown in the last column of Table \ref{tab:comparison_rectangle_swap}. It turns out that the Rectangle Loop algorithm still has a significant advantage than the swap algorithm after the running time issue is taken into account. For $p = 0.01$, the Rectangle Loop is about $31$ times more efficient than the swap algorithm. Even for $p = 0.5$, Rectangle Loop algorithm is still about $4$ times more efficient than the swap algorithm.

\begin{table}[htbp]	\begin{center}
	\begin{tabular}{c|c|c|c|}
		\hline
		\multicolumn{1}{|c|}{Method}         & Filled portion          & Number of swaps & Time per swap (/s)   \\ \hline
		\multicolumn{1}{|c|}{Rectangle Loop} & \multirow{2}{*}{$1\%$}  & $586$           & $1.18\times 10^{-5}$ \\ \cline{3-4} 
		\multicolumn{1}{|c|}{Swap}           &                         & $8$             & $3.67\times 10^{-4}$ \\ \hline
		\multicolumn{1}{|c|}{Rectangle Loop} & \multirow{2}{*}{$5\%$}  & $977$           & $5.30\times 10^{-6}$ \\ \cline{3-4} 
		\multicolumn{1}{|c|}{Swap}           &                         & $42$            & $3.52\times 10^{-5}$ \\ \hline
		\multicolumn{1}{|c|}{Rectangle Loop} & \multirow{2}{*}{$10\%$} & $1838$          & $3.23\times 10^{-6}$ \\ \cline{3-4} 
		\multicolumn{1}{|c|}{Swap}           &                         & $156$           & $1.25\times 10^{-5}$ \\ \hline
		\multicolumn{1}{|c|}{Rectangle Loop} & \multirow{2}{*}{$20\%$} & $3271$          & $2.64\times 10^{-6}$ \\ \cline{3-4} 
		\multicolumn{1}{|c|}{Swap}           &                         & $509$           & $5.68\times 10^{-6}$ \\ \hline
		\multicolumn{1}{|c|}{Rectangle Loop} & \multirow{2}{*}{$30\%$} & $4222$          & $2.10\times 10^{-6}$ \\ \cline{3-4} 
		\multicolumn{1}{|c|}{Swap}           &                         & $803$           & $5.06\times 10^{-6}$ \\ \hline
		\multicolumn{1}{|c|}{Rectangle Loop} & \multirow{2}{*}{$40\%$} & $4794$          & $1.27\times 10^{-6}$ \\ \cline{3-4} 
		\multicolumn{1}{|c|}{Swap}           &                         & $1160$          & $4.98\times 10^{-6}$ \\ \hline
		\multicolumn{1}{|c|}{Rectangle Loop}      & \multirow{2}{*}{$50\%$} & $5080$          & $1.37\times 10^{-6}$ \\ \cline{3-4} 
		\multicolumn{1}{|c|}{Swap}                            &                         & $1271$          & $5.36\times 10^{-6}$ \\ \hline
	\end{tabular}
	\caption{The comparison between swap algorithm and Rectangle Loop algorithm. Each algorithm is implemented $10000$ iterations on $100\times 100$ matrices with different filled portions. The third column records the number of successful swaps among the $10000$ iterations, the last column records the average time per swap, respectively.}
	\label{tab:comparison_rectangle_swap}
\end{center}
\end{table}
We have also used the pertubation score suggested by Strona et.al. \cite{strona2014fast} to access convergence for both algorithm.  Pertubation score of a matrix is defined by the fraction of cells differing from the corresponding ones of the initial matrix. It takes several iterations for each algorithm to stabilize around its expectation. It is shown in Figure \ref{fig:Rectangle_vs_swap} that it takes less iterations and less time for the Rectangle Loop algorithm to stabilize, suggesting a faster mixing than the swap algorithm. 
\begin{figure}[htbp]
	\includegraphics[width= \textwidth]{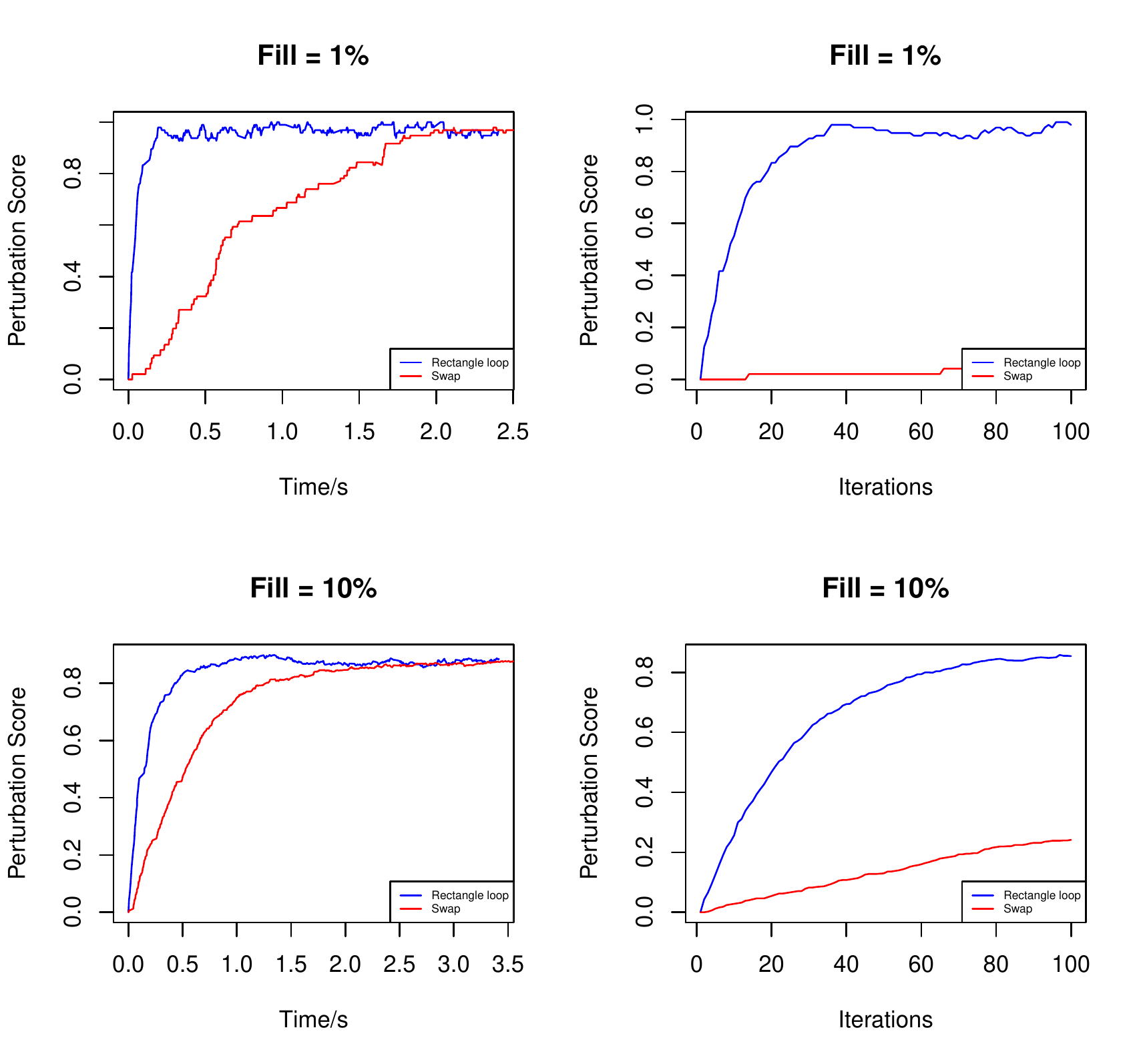}
	\caption{Comparison between swap algorithm and Rectangle Loop algorithm in mixing  $100\times100$ matrices. The left two subplots are the relationship between time and pertubation score for different filled portions. The right two subplots are the relationship between (every 100) iterations and pertubation score for different filled portions. Rectangle Loop algorithm is represented by blue curves and swap algorithm is represented by red curves.} 
	\label{fig:Rectangle_vs_swap}
\end{figure}
\subsection{Finch Data Applications}
Going back to `Darwin's Finch' dataset, we use the test statistics $\bar{S}^2$ suggested by Roberts and Stone \cite{roberts1990island} to compare the three algorithms. $\bar{S}^2$ is defined by 
\[
\bar{S}^2(A) = \frac{1}{m(m-1)}\sum_{i\neq j} s_{ij}^2,
\]
where $m$ is the number of rows of matrix $A$, $S = (s_{ij}) = A A^T$. For the finch data, $\bar{S}^2 = 45.03$. Suppose this number is too large or too small, comparing with its expectation over all the matrices having the same margins as finch data. We would like to conclude that the cooperation/competition do influence the distribution of species. To investigate this, we implemented the swap algorithm, Curveball algorithm and Rectangle Loop algorithm on the same data for $20000$ iterations, using its average as an estimator for $\bE(\barS^2)$. The results are shown in Figure \ref{fig:comparison_finch}. After $20000$ iterations, Rectangle Loop algorithm gives an estimate of $42.135$ with standard deviation $0.537$, swap algorithm gives an estimate of $42.126$ with standard deviation $0.509$, Curveball algorithm gives an estimate of $42.191$, with standard deviation $0.590$.
Therefore the observed data falls outside  the three standard deviation boundaries for all three algorithms, suggesting strong evidence that the observed occurence matrix is not just by chance. Meanwhile, both swap algorithm and Rectangle algorithm gives similar estimations and lower standard deviations, which seem to be more accurate than the Curveball algorithm. Lastly, there is a significant pattern in Figure \ref{fig:comparison_finch} that for both $\bar S^2$ and standard deviation estimation, the Rectangle Loop algorithm becomes stabilized much earilier than the swap algorithm, indicating a faster mixing.
\begin{figure}[htbp]
	\includegraphics[width= \textwidth]{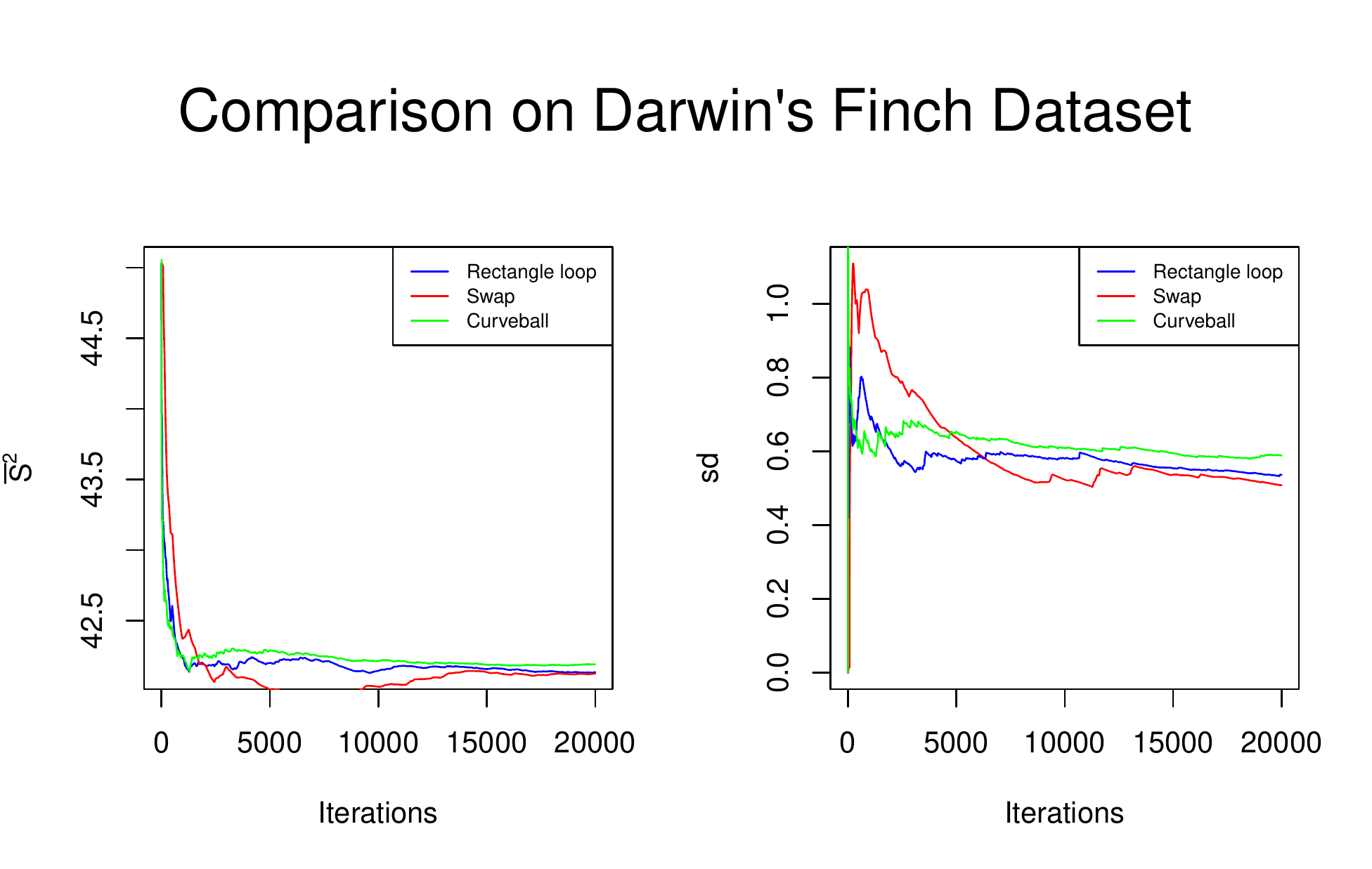}
	\caption{Comparison between swap algorithm, Curveball algorithm and Rectangle Loop algorithm. The left subplot is the relationship between the average of $\bar S^2$ with the number of iterations. The right subplot is the relationship between the standard deviation of $\bar S^2$ with the number of iterations. Blue, red and green curves represent Rectangle Loop, swap and Curveball algorithms respectively.}
	\label{fig:comparison_finch}
\end{figure}
\newpage
\section{Discussion}\label{sec:discussion}
There is a growing tendency to study the behavior of binary matrices with fixed margins in numerous scientific fields, ranging from mathematics to natural science to social science. For example, mathematicians and computer scientists are interested in the total number of configurations of given margin sums. Ecologists use the so-called \textit{occurence matrix} to model the presence/absence of species in different locations. Biologists use the binary matrix to model neuronal networks. Social scientists use the binary matrix for studying  social network features.

One of the central and difficult problems is uniformly sampling binary matrices with given margins. In this article we have developed the Rectangle Loop algorithm which is efficient, intuitive and easy to implement. Theoretically, the algorithm is superior to the classical swap algorithm in Peskun's order. In practice, the Rectangle Loop algorithm is notably more efficient than the swap approach. For a fixed number of iterations, Rectangle Loop algorithm produces $4-73$ times more successful swaps than the swap algorithm. For a fixed amount of time, Rectangle Loop algorithm still produces $4-31$ times more successful swaps than the swap algorithm. This suggests the Rectangle Loop algorithm is efficient both statistically and computationally. 

There are many other problems that remain.  From a theoretical point of view, it is important to give sharp bounds on the convergence speed of a given Markov chain. However, giving a useful running time estimate is often challenging in practical problems. It would be very interesting if the swap algorithm,  Curveball algorithm and the Rectangle Loop algorithm can be investigated analytically.  From an applied point of view, there are many factors that influence the performance of algorithms, such as  running time per step (swap algorithm is the fastest, while Curveball algorithm is the slowest),  initialization of the matrix,  size of the matrix,  ratio between row number and column numbers,  filled proportions. Our empirical studies suggest that all the factors have a significant impact on the convergence speed for all the algorithms. It would be beneficial if more numerical experiments are carried out, yielding a complete and comprehensive comparison between all the existing algorithms.

\newpage
\bibliographystyle{alpha}
\bibliography{bibliography}

\newcommand{\etalchar}[1]{$^{#1}$}
\begin{thebibliography}{CDHL05}

\bibitem[BC89]{besag1989generalized}
Julian Besag and Peter Clifford.
\newblock Generalized monte carlo significance tests.
\newblock {\em Biometrika}, 76(4):633--642, 1989.

\bibitem[CDHL05]{chen2005sequential}
Yuguo Chen, Persi Diaconis, Susan~P Holmes, and Jun~S Liu.
\newblock Sequential monte carlo methods for statistical analysis of tables.
\newblock {\em Journal of the American Statistical Association},
  100(469):109--120, 2005.

\bibitem[CK18]{carstens2018speeding}
Corrie~Jacobien Carstens and Pieter Kleer.
\newblock Speeding up switch markov chains for sampling bipartite graphs with
  given degree sequence.
\newblock In {\em Approximation, Randomization, and Combinatorial Optimization.
  Algorithms and Techniques (APPROX/RANDOM 2018)}. Schloss
  Dagstuhl-Leibniz-Zentrum fuer Informatik, 2018.

\bibitem[DG95]{diaconis1995rectangular}
Persi Diaconis and Anil Gangolli.
\newblock Rectangular arrays with fixed margins.
\newblock In {\em Discrete probability and algorithms}, pages 15--41. Springer,
  1995.

\bibitem[DS{\etalchar{+}}98]{diaconis1998algebraic}
Persi Diaconis, Bernd Sturmfels, et~al.
\newblock Algebraic algorithms for sampling from conditional distributions.
\newblock {\em The Annals of statistics}, 26(1):363--397, 1998.

\bibitem[G{\etalchar{+}}57]{gale1957theorem}
David Gale et~al.
\newblock A theorem on flows in networks.
\newblock {\em Pacific J. Math}, 7(2):1073--1082, 1957.

\bibitem[HJ96]{holmes1996uniform}
RB~Holmes and LK~Jones.
\newblock On uniform generation of two-way tables with fixed margins and the
  conditional volume test of diaconis and efron.
\newblock {\em The Annals of Statistics}, pages 64--68, 1996.

\bibitem[HM13]{harrison2013importance}
Matthew~T Harrison and Jeffrey~W Miller.
\newblock Importance sampling for weighted binary random matrices with
  specified margins.
\newblock {\em arXiv preprint arXiv:1301.3928}, 2013.

\bibitem[KTV99]{kannan1999simple}
Ravi Kannan, Prasad Tetali, and Santosh Vempala.
\newblock Simple markov-chain algorithms for generating bipartite graphs and
  tournaments.
\newblock {\em Random Structures \& Algorithms}, 14(4):293--308, 1999.

\bibitem[MP04]{miklos2004randomization}
Istv{\'a}n Mikl{\'o}s and J{\'a}nos Podani.
\newblock Randomization of presence--absence matrices: comments and new
  algorithms.
\newblock {\em Ecology}, 85(1):86--92, 2004.

\bibitem[Pes73]{peskun1973optimum}
Peter~H Peskun.
\newblock Optimum monte-carlo sampling using markov chains.
\newblock {\em Biometrika}, 60(3):607--612, 1973.

\bibitem[Rec18]{rechner2018markov}
Steffen Rechner.
\newblock Markov chain monte carlo algorithms for the uniform sampling of
  combinatorial objects.
\newblock 2018.

\bibitem[RJB96]{rao1996markov}
A~Ramachandra Rao, Rabindranath Jana, and Suraj Bandyopadhyay.
\newblock A markov chain monte carlo method for generating random (0,
  1)-matrices with given marginals.
\newblock {\em Sankhy{\=a}: The Indian Journal of Statistics, Series A}, pages
  225--242, 1996.

\bibitem[RS90]{roberts1990island}
Alan Roberts and Lewis Stone.
\newblock Island-sharing by archipelago species.
\newblock {\em Oecologia}, 83(4):560--567, 1990.

\bibitem[Rys09]{ryser2009combinatorial}
Herbert~J Ryser.
\newblock Combinatorial properties of matrices of zeros and ones.
\newblock In {\em Classic Papers in Combinatorics}, pages 269--275. Springer,
  2009.

\bibitem[SNB{\etalchar{+}}14]{strona2014fast}
Giovanni Strona, Domenico Nappo, Francesco Boccacci, Simone Fattorini, and
  Jesus San-Miguel-Ayanz.
\newblock A fast and unbiased procedure to randomize ecological binary matrices
  with fixed row and column totals.
\newblock {\em Nature communications}, 5:4114, 2014.

\bibitem[Sni91]{snijders1991enumeration}
Tom~AB Snijders.
\newblock Enumeration and simulation methods for 0--1 matrices with given
  marginals.
\newblock {\em Psychometrika}, 56(3):397--417, 1991.

\bibitem[Ver08]{verhelst2008efficient}
Norman~D Verhelst.
\newblock An efficient mcmc algorithm to sample binary matrices with fixed
  marginals.
\newblock {\em Psychometrika}, 73(4):705, 2008.

\end{thebibliography}
\end{document}